\documentclass[11pt,a4paper]{article}
\usepackage[utf8]{inputenc}
\usepackage{geometry}
\usepackage[english]{babel}
\usepackage[dvipsnames]{xcolor}
\usepackage{amssymb}
\usepackage{amsmath}
\usepackage{amsfonts}
\usepackage{algorithm}
\usepackage{algorithmicx}
\usepackage{enumitem}
\usepackage{algpseudocode}
\usepackage{float}
\usepackage{graphics}
\usepackage{graphicx}
\usepackage{amsmath}

\newcommand{\probl}[3]{
  \begin{flushleft}
    \fbox{
      \begin{minipage}{\textwidth}
        \noindent {\sc #1}\\
        {\bf Input:} #2\\
        {\bf Output:} #3
      \end{minipage}}
    \medskip
  \end{flushleft}
}

\newcommand{\exemple}{\newline \indent {\textit{Exemple: }}}

\date{}

\usepackage{pdfpages}

\usepackage{latexsym} %Des symboles de maths en +
\usepackage{theorem} %pour changer les styles dans les theoremes.

\usepackage{thm-restate} %Réénoncer le thm exactement comme précédemment

\usepackage{hyperref}
\usepackage[capitalize]{cleveref}

\usepackage{tikz}

\usetikzlibrary{positioning} 
\usetikzlibrary{shapes}
\tikzstyle{Template_A}=[draw,circle,thick, minimum size = 0.5cm, color=black,inner sep=0pt]
\tikzstyle{Template_B}=[color=black,draw]
\tikzstyle{vertex}=[draw,circle,fill=black, color=black,inner sep=1.5pt]
\tikzstyle{vertex2}=[draw,circle,fill,inner sep=2pt]
\tikzstyle{vertex3}=[draw,circle,fill,inner sep=0.5pt]
\tikzstyle{S&L}=[draw,circle,minimum height=0.5cm]
\tikzstyle{labeled}=[draw,circle,inner sep=0.35pt]
\tikzstyle{boite}=[draw,rectangle,rounded corners=3pt]

\newtheorem{DE}{Definition}[section]
\newtheorem{definition}[DE]{Definition}
\newtheorem{theorem}[DE]{Theorem}

\newtheorem{lemma}[DE]{Lemma}

\newcounter{claim}
\counterwithin{claim}{theorem}
\newenvironment{proof}[1][]%
{\noindent{\it Proof. }{#1}{}}{\qed\vspace{2ex}}
\newenvironment{claim}[1][]%
{\refstepcounter{claim}\vspace{1ex} {(\it\arabic{section}.\arabic{theorem}.\arabic{claim}) {#1}{}}\it}{\vspace{2ex}}
\newenvironment{proofclaim}[1][]%
{\noindent {}{#1}}{This proves~(\arabic{section}.\arabic{theorem}.\arabic{claim}).\vspace{2ex}}

\newcommand {\sm} {\setminus}

\newcommand{\qed}{\relax\ifmmode\hskip2em\Box\else\unskip\nobreak\hfill$\Box$\fi}

\newcommand {\kmax} {k_{m}}
\newcommand {\ksum} {k_{s}}
\newcommand {\elementset} {\mathcal{E}}
\newcommand {\solutionset} {\mathcal{U}}
\newcommand {\inputset} {\mathcal{I}}

\newcommand {\rtprt}[2] {\gamma(#1,#2)}
\newcommand {\rtdist}[2] {\delta(#1,#2)}

\author{Julien Baste\thanks{Univ. Lille, CNRS, Centrale Lille, UMR 9189 CRIStAL, F-59000 Lille, France}, Cléophée Robin\thanks{Université Paris Cité, CNRS, IRIF, F-75013, Paris, France} and Marie-Emilie Voge\footnotemark[1] }
\title{On the diversity problems parameterized by the size of the solutions}

\begin{document}

\maketitle

\begin{abstract}
  Diversity optimization seeks multiple high-quality solutions that are sufficiently different from one another, providing a richer representation of the solution space than a single optimum while avoiding the prohibitive cost of complete enumeration.
  In this work, we introduce the notion of natural diversity, a general condition that connects a combinatorial problem $\Pi$ and a diversity measure $\texttt{dist}$.
  We show that if a pair $(\Pi,\texttt{dist})$ is naturally diverse and we have in hand a completion algorithm that, from a partial solution of $\Pi$, can complete it into a solution of $\Pi$, then we can produce a parameterized algorithm solving the diversity variant of $\Pi$ with regard to the distance $\texttt{dist}$ when parameterized by the size of the solutions and the number of expected solutions.
  Furthermore, we show that several widely used diversity measures, including pairwise disjointness, Hamming distance, Jaccard distance, and the Otsuka–Ochiai coefficient (under both minimum and sum aggregation), satisfy the natural diversity property for problems in which all the solutions have the same size.
  Finally, we demonstrate the applicability of our framework by deriving fixed-parameter algorithms for diverse variants of Minimum Vertex Cover and Minimum Steiner Tree. Our results broaden the scope of parameterized diversity algorithms by accommodating natural solution-size parameterizations and a wider class of diversity measures.
\end{abstract}

\section{Introduction}

The \textit{diversity problems} is a set of problems where the goal is to find a fixed number of solutions guaranteed to be distant with regard to a given distance.
These problems are a compromise between searching for a unique solution that is faster but  provides really little information on the solution space, and enumerating all solutions that provide all the information on the solution space at the cost of the running time.

Solving diversity problems has become a very active topic in the last years, in particular in artificial intelligence~\cite{InGaStTa2020,HaKoKuLeOt2022,HaKoKuOt2021,IcIw2024} and in algorithmic~\cite{FoGoJaPhSa2020,FoGoPaPhSa2021}. It is also of high interest in multiple fields like
Evolutionary Computation~\cite{WiOp2003,GaBePhSc2018},
Constraint Programming~\cite{HeHnSuWa2005,PeTr2015},
Mixed Integer Programming~\cite{DaWo2009,AhMeTr2024},
SAT~\cite{NaSaSi2011},
Automata~\cite{ArFeOlWo2023},
Submodular function optimization~\cite{DoGuNeNe2023,NeBoNe2021},
Social choice theory~\cite{ArFeLoOlWo2021},
and
Answer set programming~\cite{EiErErFi2013}.

In a recent work, \cite{AuBeGoLiSr2025} provides randomized algorithms running in time $O(r^3 \cdot 2^{c\cdot n})$, where $0 < c < 1$ is a constant depending on the base problem, $r$ is the number of searched solutions and $n$ the size of the input, to solve the diversity variant of $24$ problems including \textsc{Vertex Cover}, \textsc{Maximum Independent Set}, and \textsc{Feedback Vertex Set}, where the distance is either the minimum or the sum of the Hamming distance between each distinct pair of solutions.

The\textit{ diverse }variant, as defined before, of a problem can be seen as a generalization of the original problem.
In this way the diverse variant is at least as hard as the original version.
Intuitively, one has to set the number of searched solutions to one.
In~\cite{CrRo2002}, the authors show that if finding one solution for the \textsc{$2$-SAT} problem is well known to be polynomial, finding two solutions that maximize the Hamming distance between them is \texttt{PolyApx}-complete.
This shows that a complexity gap can be expected between the problem and its diversity variant.
A natural question is then to determine in which context the complexity of the diversity variant can be connected to the complexity of the original problem.

A first answer to this problem has been provided in~\cite{fvBaFeJaMaOlPhRo2022} in the context of graph problems and the parameterized complexity theory.
In this paper, the authors construct a meta-algorithm that takes as input a specific algorithm for the initial problem and outputs an algorithm that solves the diversity variant.
In the case of~\cite{fvBaFeJaMaOlPhRo2022}, the specified algorithm should be an algorithm parameterized by treewidth, the latter being a structural graph parameter that can be seen as a topological measure between the input graph and a tree.
The outputted algorithm is then an algorithm parameterized by the treewidth of the input graph and the number of searched solutions solving the diversity version when the used distance is the sum of the hamming distances between solutions.
In~\cite{DrMa2024}, the authors provide the same kind of results with clique-width as a parameter instead of treewidth.

\paragraph{Our results.}

While \cite{fvBaFeJaMaOlPhRo2022} and \cite{DrMa2024} provide very powerful parameterized meta-algorithms, these meta-algorithms are restricted, by design, to structural graph parameters.
In this paper, we show that the same kind of meta-algorithm can be constructed when the parameter is the size of the solution.
We want to mention that this later parameter is often considered more natural and moreover is not restricted to graph algorithms.

Moreover, in \cite{fvBaFeJaMaOlPhRo2022} the results are strongly based on the fact that the diversity distances considered can be incrementally computed. In this paper, we are freeing ourselves from this constraint by providing a global requirement on the pair problem/distance.
This global requirement, that we named natural diversity, and formally define in Section~\ref{sec:prelim}, states that completing a partial solution with elements that are not used by any other solutions provides a result that is at least as diverse as any other completion.

In Section~\ref{sec:prelim} we introduce the notation and definition used in the paper, in Section~\ref{sec:algo} we present our main result together with its proof, in Section~\ref{sec:dist} we show some distances compatible with our meta-algorithm, in Section~\ref{sec:ex} we provide some application examples and finish in Section~\ref{sec:conclusion} by a conclusion.

\section{Preliminaries}
\label{sec:prelim}

\paragraph{Sets.}
We denote by $\mathbb{N}$ the set of non-negative integers and by $\mathbb{R}$ the set of all real numbers.
Given two integers $a$ and $b$, we denote by  $[a,b]$ the set of every integer $k$ such that $a \leq k \leq b$.

\paragraph{Graphs.}
We denote by $\mathcal{G}$ the set of all graphs.
Given a graph $G$, we denote by $V(G)$ its set of vertices and $E(G)$ its set of edges.
All the graphs we consider in this paper are simple, non-oriented, and without loops, thus we safely consider each edge as a set of two distinct vertices.
Given a graph $G$ and a set $S \subseteq E(G)$, we denote by $V(S)$ the set $\bigcup_{e\in S}e$ and by $G[S]$ the graph $(V(S),S)$
Given a graph $G$ and a set $X \subseteq V(G)$, we denote by $G[X]$ the subgraph of $G$ induced by $X$.

Given a graph $G=(V,E)$ and two vertices $x$ and $y$ a \textit{$\{x,y\}$-path} $p$ of $G$ is a subset of $E$, such that the graph $G[p]$ is connected  with all vertices of having degree 2 except $x$ and $y$ that have degree 1.

Through this section, in order to illustrate the concepts that we will introduce in the paper, we use a classical parameterized problem, namely \textsc{Vertex Cover}.

\probl{Vertex Cover}
{A graph $G$ and an integer $k$.}
{A set $S \subseteq V(G)$ of size at most $k$ such that for each $e \in E(G)$, $S \cap e \not = \emptyset$ or the correct statement that such a set does not exist.}

Given a graph $G$, a set $S \subseteq V(G)$ is a \emph{vertex cover} of $G$ if for each $e \in E(G)$, $S \cap e \not = \emptyset$.

\paragraph{Properties.}
Let $\inputset$ be a set of inputs, $\elementset$ be a set and $\solutionset = 2^{\elementset}$.
We focus on properties $\Pi$ over the set $\inputset \times \solutionset$.
If $(I,S) \in \inputset \times \solutionset$ satisfy $\Pi$, we write $(I,S) \vdash \Pi$.
In this case, we call $\inputset$ the \emph{input space} of $\Pi$ and $\mathcal{U}$ the \emph{solution space} of $\Pi$.
We always assume that a property $\Pi$ arrives with its input space $\inputset_{\Pi}$ and its solution space $\mathcal{U}_{\Pi}$.
In this paper, we will consider only properties $\Pi$ such that, for each $(I,S) \in \inputset_{\Pi}\times \mathcal{U}_{\Pi}$, we can verify   in polynomial time that $(I,S) \vdash \Pi$.
We say that $\Pi$ is a \emph{fixed size property} if for each $(I,S_1,S_2) \in \inputset_{\pi}\times \mathcal{U}_{\pi}\times \mathcal{U}_{\pi}$, such that $(I,S_1) \vdash \Pi$ and $(I,S_2) \vdash \Pi$, we have $|S_1| = |S_2|$.
\exemple Let $\mathcal{I} = \mathcal{G} \times \mathbb{N}$, let $\mathcal{E}$ be the set of all vertices of graphs in $\mathcal{G}$ and $\mathcal{U}$ the set of all subsets of $\mathcal{E}$. We can consider the property \emph{being a vertex cover of correct size} $\Pi$ such that $(I = (G,k),S) \in \inputset \times \solutionset$ satisfy $\Pi$ if and only if $S$ is a vertex cover of $G$ of size at most $k$.
Moreover, for a property $\Pi'$ such that $(G,S) \in \mathcal{G} \times \solutionset$ satisfies $\Pi'$ if and only if  $S$ is a vertex cover of $G$ of minimum size, we can say that $\Pi'$ is a fixed size property.

\paragraph{Input and output} Given an algorithm, an input (resp. output) is \textit{valid} if it respects the specification of the input (resp. output) in the definition of the algorithm.

\paragraph{Distance functions.}
A \emph{$r$-distance function} over a set $\solutionset$, $\texttt{dist}$, is a function computable in polynomial time from $\mathcal{U}^r \to \mathbb{N}$.
Intuitively, these functions will evaluate how diverse a set of $r$ solutions is.
When $r$ is clear from the context, we also name it \emph{distance function}.
\exemple We consider the $r$-distance function $\texttt{r-disjoint}: \solutionset^r \to \{0,1\}$, formally defined in Section~\ref{sec:dist}, that returns $1$ if all the input sets are pairwise disjoint and $0$ if at least two sets overlap.

\paragraph{Natural diversity}
The results proved throughout this paper are about pairs  $(\Pi,\texttt{dist})$ with $\Pi$ a property and $\texttt{dist}$ a distance function over  $\mathcal{U}_{\Pi}$.
The results will require that this pair is naturally diverse as defined in Definition~\ref{def:natdiv}.
We discuss in Section~\ref{sec:dist} some conditions to be respected such that the pair $(\Pi,\texttt{dist})$ is naturally diverse.

\begin{definition}[naturally diverse]
  \label{def:natdiv}
  Given a property $\Pi$ and a $r$-distance function $\texttt{dist}$ over $\mathcal{U}_{\Pi}$, we say that $(\Pi,\texttt{dist})$ is \emph{naturally diverse}
  if
  
  \begin{itemize}
  \item for each $I \subseteq \inputset_{\Pi}$,
  \item for each $(S_1,\ldots, S_r, X, Y) \in {\mathcal{U}_{\Pi}}^{r+2}$
    \begin{itemize}
    \item such that $S_i \cap Y = \emptyset$ for each $i \in [1,r]$,
    \end{itemize}
  \item for each $i \in [1,r]$
    \begin{itemize}
    \item such that $X \subseteq S_i$, $(I,S_i) \vdash \Pi$, and $(I,(S_i \setminus X) \cup Y) \vdash \Pi$,
    \end{itemize}
  \end{itemize}
  we have that:
  $$\texttt{dist}(S_1, \ldots, S_{i-1}, S_i, S_{i+1}, \ldots S_r) \leq \texttt{dist}(S_1, \ldots, S_{i-1}, (S_i \setminus X) \cup Y, S_{i+1}, \ldots S_r).$$
\end{definition}

Roughly speaking, this definition transcripts the idea that when one is looking for sets that are diverse, completing a partial solution (here $S_i \setminus X)$ with elements that are not used in any other solutions (here $Y$) will produce a solution that is at least as diverse as completing it with any other elements (here $X$).

\begin{lemma}
  \label{rk:natdiv}
  Let $\Pi_1, \ldots, \Pi_r$ be $r$ properties and $\texttt{dist}$ be a $r$-distance function  over $\mathcal{U}_{\Pi_1} = \ldots = \mathcal{U}_{\Pi_r}$ such that for each $i \in [1,r]$, we have that $(\Pi,\texttt{dist})$ is \emph{naturally diverse}.
  Let $(I_1, \ldots, I_r) \in \inputset_{\Pi_1} \times \ldots \times \inputset_{\Pi_r}$, 
  $(S_1, \ldots, S_r) \in \solutionset_{\Pi_1} \times \ldots \times \solutionset_{\Pi_r}$, and
  $(S'_1, \ldots, S'_r) \in \solutionset_{\Pi_1} \times \ldots \times \solutionset_{\Pi_r}$ such that
  for each $i \in [1,r]$, we have $S_i \subseteq S'_i$, $(I_i,S_i) \vdash \Pi_i$ and $(I_i,S'_i) \vdash \Pi_i$.
  
  In this context we have that 
  $\texttt{dist}(S'_1, \ldots, S'_r) \leq \texttt{dist}(S_1, \ldots, S_r).$
\end{lemma}

\begin{proof}
  Let, for each $i \in [1,r]$, $X_i = S'_i \setminus S_i$ and $Y_i = \emptyset$.
  We have that, for each $i \in [1,r]$, $S_i = (S'_i \setminus X_i) \cup Y_i$ and for each $j \in [1,r]$, $S'_j \cap Y_i = \emptyset$. Combining this with Definition~\ref{def:natdiv} provides the remark.
\end{proof}

\paragraph{Diversity problems.}
Let $r$ be an integer,  $\inputset_1, \ldots, \inputset_r$ and $\elementset$ be $r+1$ sets, $\solutionset = 2^{\elementset}$,  \texttt{dist} be a $r$-distance functions over $\solutionset$, and
for each $i \in [1,r]$, $\Pi_i$ be a property over the input space $\inputset_i$ and the solution space $\solutionset$.
The main problem addressed in this paper is \textsc{$\texttt{dist}$-Diverse-($\Pi_1$,$\ldots$, $\Pi_r$)} defined as follows:

\probl{$\texttt{dist}$-Diverse-($\Pi_1$,$\ldots$, $\Pi_r$)}
{$r$ inputs $(I_1, \ldots, I_r) \in \inputset_1 \times \ldots \times \inputset_r$ and an integer $d$.}
{$r$ sets $S_1, \ldots, S_r \in \mathcal{U}$ such that for each $i \in [1,r]$, $(I_i,S_i) \vdash \Pi_i$ and $\texttt{dist}(S_1, \ldots, S_r) \geq d$, or the correct statement that such a set does not exist.}

In order to solve the previous problem, we always assume that for each property $\Pi$ involved in the \textsc{$\texttt{dist}$-Diverse-($\Pi_1$,$\ldots$, $\Pi_r$)} problem, there exists an algorithm solving the $\Pi$-completion problem, defined as follows.

\probl{$\Pi$-completion}
{An input $I \in \inputset_{\Pi}$ and two sets $A$ and $B$ in $\mathcal{U}_{\Pi}$ such that $A \subseteq B$.}
{A set $S \in \mathcal{U}_{\Pi}\setminus B$ such that $(I,A \cup S) \vdash \Pi$  or a correct statement that such a set does not exist.}

Intuitively, the \textsc{$\Pi$-completion} problem can be seen as follows: Assuming that we already have processed the vertices contained in the set $B$ and decided that from them, $A$ is part of the solution and the remaining is not, can we complete $A$ with elements not in $B$  to a solution of $\Pi$.

\section{General algorithm}
\label{sec:algo}

With all the definitions stated in the previous section, we are now ready to state the main theorem of the paper.

\begin{restatable}{theorem}{main}
  \label{th:main}
  Let $r$ be an integer,  $\inputset_1, \ldots, \inputset_r$ and $\elementset$ be $r+1$ sets, $\solutionset = 2^{\elementset}$,  \texttt{dist} be a $r$-distance function over $\solutionset$, and
  for each $i \in [1,r]$, $\Pi_i$ be a property over the input space $\inputset_i$ and the solution space $\solutionset$
  such that:
  \begin{itemize}
  \item  there exists an algorithm $\mathcal{A}_i$ that solves \textsc{$\Pi_i$-completion} in time $f_i(|I_i|,|A|,|B|)$ on input $(I_i,A,B)$ and 
  \item $(\Pi_i,\texttt{dist})$ is naturally diverse.
  \end{itemize}
  Then, there exists an algorithm that solves, for every input $(I_1, \ldots, I_r,d) \in \inputset_1 \times \ldots \times \inputset_r \times \mathbb{N}$,  \textsc{\texttt{dist}-Diverse-($\Pi_1$,$\ldots$, $\Pi_r$)} in time 
  $$O\big((r \cdot \rtprt{I}{\kmax} + \rtdist{r}{\kmax} + F + 2^{r \cdot \kmax}\cdot (r + \kmax)^{O(1)}) \cdot 2^{r \cdot \kmax \cdot \ksum}\big) $$
  where
  \begin{itemize}
  \item $k_i$ is the maximum size of a solution $S$ such that $(I_i,S) \vdash \Pi_i$ for each $i \in [1,r]$,
    $\ksum = \sum_{i \in [1,r]}k_i$ and $\kmax = \max_{i \in [1,r]} k_i$.
  \item $F = \max_{i\in[1,r]}\max_{A \subseteq B \subseteq \mathcal{U}}f_i(|I_i|,|A|,|B|)$, 
  \item  $\rtprt{I}{q}$ denotes the maximum runtime for testing $(I_i,S) \vdash \Pi_i$ over every $i \in [1,r]$ and $(I_i,S) \in \inputset_i \times \solutionset$ such that  $|S| \leq q$.
  \item $\rtdist{p}{q}$ denotes the maximum runtime for computing $\texttt{dist}(S_1, \ldots, S_p)$ where for each $i \in [1,p]$, we have $|S_i| \leq q$.
  \end{itemize}
\end{restatable}

We start by giving some intuition of the proof of the theorem before going through the details.

We produce a branching algorithm (Algorithm~\ref{alg:subD}) that keeps at each node a potential partial solution $(S_1, \ldots, S_r)$ and a set  $R$ of elements of $\elementset$ already considered.
Note that the algorithm will maintain that $S_i \subseteq R$ for each $i \in [1,r]$.
Given a node of the branching, the algorithm solves the \textsc{$\Pi_{i_0}$-completion} problem, with input $(I_{i_0},S_{i_0},R)$, for some well chosen $i_0 \in [1,r]$.
This produces a set $X$ disjoint from $R$.
Then, the algorithm branches on each possible $(S_1 \cup X_1, \ldots, S_r \cup X_r)$ and $R \cup X$ where $X_i \subseteq X$ for each $i \in [1,r]$.
Note that $\sum_{i \in [1,r]}|S_i| < \sum_{i \in [1,r]}|S_i \cup X_i|$ exept when all $X_i = \emptyset$, $i \in [1,r]$.
For this later case, as $(\Pi_{i_0}, \texttt{dist})$ is naturally diverse, we ensure that the branch where $X_{i_0} = X$ and $X_i = \emptyset$, for each $i \in [1,r] \setminus \{i_0\}$,  produces a result that is at least as good.
Thanks to this, the algorithm will not consider this branch.
As, for each $i \in [1,r]$, $k_i$ is the maximum size of a solution $S$ such that $(I_i,S) \vdash \Pi_i$. The depth of the branching tree is at most $\ksum = \sum_{i \in [1,r]}k_i$ and each step runs \textsc{$\Pi_i$-completion} in time $f_i(|I_i|,|S_i|,|R|)$. There are at most $2^{\ksum}$ branches. The runtime follows.

For proving formally Theorem~\ref{th:main}, we design Algorithm~\ref{alg:subD}, we prove its correctness  in Lemma~\ref{lemma:correctness} and its running time in Lemma~\ref{lemma:runtime} and use both lemmas to conclude.

\begin{algorithm}[b!]
  \caption{$\texttt{dist}$-sub-Diverse-($\Pi_1$,$\ldots$, $\Pi_r$) : \emph{sDiv}}
  \label{alg:subD}
  \begin{algorithmic}[1]
    \Require $I = (I_1, \ldots, I_r) \in \inputset_1 \times \ldots \times \inputset_r$, $d \in \mathbb{N}$, $(S_1,S_2,\dots, S_r,R) \in \mathcal{U}^{r+1}$ such that $\forall i \in [1,r]: S_i \subseteq R$
    \vspace{2ex}
    
    \Ensure A tuple $(S'_1, \ldots S'_r,R') \in \mathcal{U}^{r+1}$ such that $R \subseteq R'$, $\forall i \in [1,r]: S_i \subseteq S'_i \subseteq R'$,   $(S'_i \setminus S_i) \cap R = \emptyset$, $(I_i,S'_i) \vdash \Pi_i$ and $\texttt{dist}(S'_1, \ldots, S'_r) \geq d$, or a correct statement that such a tuple does not exist.
    \vspace{3ex}
    
    \If { $\forall i\in [1,r]:$ $(I_i,S_i) \vdash \Pi_i$ and $\texttt{dist}(S_1, \ldots, S_r) \geq d$}
    \State \textbf{Return} $(S_1,S_2,\dots, S_r, R)$
    \ElsIf { $\forall i \in [1,r]: (I_i,S_i) \vdash \Pi_i$ and $\texttt{dist}(S_1, \ldots, S_r) < d,$}
    \State \textbf{Return} \texttt{No}
    
    \Else
    \State Let $i_0 \in [1,r]$ such that  $(I_{i_0},S_{i_0}) \not \vdash \Pi_{i_0}$
    \State $X\gets \mathcal{A}_{i_0} ({I_{i_0}},  S_{i_0}, R)$
    \If {$X$ is not \texttt{No}}
    \For  { each $(X_1,\ldots, X_r) \subseteq X^r$}
    \If {$\bigcup_{i \in [1,r]} X_i \neq \emptyset$}
    \State ${\mathcal{S}} \gets (S_1\cup X_1,\dots, S_r\cup X_r, R \cup X)$
    \State Result $\gets sDiv(I, d, {\mathcal{S}})$
    \If { Result is not \texttt{No}}
    \State \textbf{Return } Result
    
    \EndIf               
    \EndIf
    \EndFor
    \EndIf    
    \State \textbf{Return}  \texttt{No}
    \EndIf

  \end{algorithmic}
\end{algorithm}

\begin{lemma}[Correctness of Algorithm~\ref{alg:subD}]
  \label{lemma:correctness}
  Let $r$ be an integer,  $\inputset_1, \ldots, \inputset_r$ and $\elementset$ be $r+1$ sets, $\solutionset = 2^{\elementset}$,  \texttt{dist} be a $r$-distance function over $\solutionset$, and
  for each $i \in [1,r]$, $\Pi_i$ be a property over the input space $\inputset_i$ and the solution space $\solutionset$
  such that:
  \begin{itemize}
  \item there exists an algorithm $\mathcal{A}_i$ that solves \textsc{$\Pi_i$-completion} in time $f_i(|I_i|,|A|,|B|)$ on input $(I_i,A,B)$ and 
  \item $(\Pi_i,\texttt{dist})$ is naturally diverse.
  \end{itemize}
  then, Algorithm~\ref{alg:subD}, on a valid input produces a valid output.
\end{lemma}

\begin{proof}
  Assume the conditions of the lemma.
  Moreover, for any integer $p$, given $(A_1, \ldots, A_p)$ and $(B_1, \ldots, B_p)$, two elements of $\mathcal{U}^{p}$, we say that 
  $(A_1, \ldots, A_p)$ is smaller than $(B_1, \ldots, B_p)$, and write
  $(A_1, \ldots, A_p) \prec_p (B_1, \ldots, B_p)$ if for each $i \in [1,p]$, we have $A_i \subseteq B_i$ and there exists $i \in [1,p]$ such that $A_i \not = B_i$.

  We divide the proof of the lemma into two parts.
  
  \begin{claim}\label{claim:yes}
    If Algorithm~\ref{alg:subD} outputs a tuple $(S'_1, \ldots, S'_r,R')$  then $(S'_1, \ldots S'_r,R') \in \mathcal{U}^{r+1}$, $R \subseteq R'$, $\forall i \in [1,r]: S_i \subseteq S'_i \subseteq R'$, $(S'_i \setminus S_i) \cap R = \emptyset$,   $(I_i,S'_i) \vdash \Pi_i$ and $\texttt{dist}(S'_1, \ldots, S'_r) \geq d$.
  \end{claim}
  
  \begin{proofclaim}
    
    We provide a proof of correctness using a maximal counter example.
    Let $(I = (I_1, \ldots, I_r),d, (S^*_1, \ldots, S^*_r,R^*))$ be a valid input of Algorithm~\ref{alg:subD}.
    Suppose that Algorithm~\ref{alg:subD} returns a tuple ${\cal S}^*=(\hat{S}^*_1, \ldots, \hat{S}^*_r,\hat{R}^*)$ and assume that it is not a valid output.
    Over all the recursive call, done at line 12 of the algorithm, including the initial call, of Algorithm~\ref{alg:subD} that returns a tuple ${\cal S}=(\hat{S}_1, \ldots, \hat{S}_r,\hat{R})$, we consider the one where $\mathcal{S}$ is not a valid output and such that the part $(S_1, \ldots, S_r,R)$ of the input $(I,d, (S_1, \ldots, S_r,R))$ is maximal with regard to $\prec_{r+1}$ .

    Observe that Algorithm~\ref{alg:subD} returns a tuple $(\hat{S}_1, \ldots, \hat{S}_r, \hat{R})$ at line 2 or line 14.
    
    At line 2, the tuple returned is  $(S_1,\dots, S_r,R)$. By definition, $(S_1, \ldots S_r,R) \in \mathcal{U}^{r+1}$, $R\subseteq R$, $\forall i \in [1,r]:  S_i \subseteq R$, and $(S_i \setminus S_i) \cap R = \emptyset$. In addition, thanks to the condition of line 1, we have $(I_i,S_i) \vdash \Pi_i$ and $\texttt{dist}(S_1, \ldots, S_r) \geq d$ for each $i \in [1,r]$ and so the solution is correct, that is not possible by assumption.
    
    It remains the case where Algorithm~\ref{alg:subD} returns a tuple ${\cal S}=(\hat{S}_1, \ldots, \hat{S}_r,\hat{R})$ at  line 14. The tuple ${\cal S}$ is the output given at line 12 by Algorithm~\ref{alg:subD} on the input tuple  $(I,d,(S_1\cup X_1,\dots, S_r\cup X_r, R \cup X))$. 
    Line 10 of Algorithm~\ref{alg:subD}  ensure that   $(I,d, (S_1, \ldots, S_r,R)) \prec_{r+1} (I,d, (S_1\cup X_1,\dots, S_r\cup X_r, R \cup X)$. As $(I,d, (S_1, \ldots, S_r,R))$ is the input of our maximal counter example the value returned at line 12  is a valid output for input  $(I,d, (S_1\cup X_1,\dots, S_r\cup X_r, R \cup X)$. 
    Observe that $R\cup X\subseteq \hat{R}$ and so $R \subseteq  \hat{R}$.
    Moreover, as for each $i \in [1,r]$, $(\hat{S}_i \setminus (S_i \cup X_i)) \cap (R \cup X) = \emptyset$,
    we have $(\hat{S}_i \setminus (S_i \cup X_i)) \cap X = \emptyset$, implying $\hat{S}_i \cap X \subseteq (S_i \cup X_i) \cap X$, thus $\hat{S}_i \cap X = X_i$.
    In addition,$\forall i \in [i,r]$, $(I_i,\hat{S}_i) \vdash \Pi_i$ and   $\texttt{dist}(\hat{S}_1, \ldots, \hat{S}_r) \geq d$. Since $S_i\cup X_i \subseteq \hat{S}_i\subseteq \hat{R}$ we have $ S_i \subseteq \hat{S}_i\subseteq \hat{R}$.  By definition of $\mathcal{A}_{i_0}$, $R\cap X=\emptyset$. Hence, since $\hat{S}_i\sm S_i\cap R=(\hat{S}_i\sm (S_i\cup X_i)\cap R) \cup (\hat{S}_i\cap R\cap X_i$, we have $\hat{S}_i\sm S_i\cap R=(\hat{S}_i\sm (S_i\cup X_i)\cap R)$. Since $(\hat{S}_i\setminus (S_i\cup X_i) \cap (R\cup X) = \emptyset$,  $(\hat{S}_i\setminus (S_i\cup X_i) \cap R = \emptyset$ and so $\hat{S}_i\sm S_i\cap R=\emptyset$.
    Moreover, as for each $i \in [1,r]$, $S_i \subseteq R \subseteq \elementset$, $X_i \subseteq X \subseteq \elementset$, we have that $(S_1\cup X_1,\dots, S_r\cup X_r, R \cup X)\in \mathcal{U}^{r+1}$.
    Hence  ${\cal S}=(\hat{S}_1, \ldots, \hat{S}_r,\hat{R})$ is a valid output, a contradiction for being a maximal counter example. 
    
    Thus, if  Algorithm~\ref{alg:subD} returns a tuple ${\cal S}^*=(\hat{S}^*_1, \ldots, \hat{S}^*_r,\hat{R}^*)$, then this tuple is a valid output.
  \end{proofclaim}
  
  \begin{claim}
    \label{claim:no}
    If Algorithm~\ref{alg:subD} outputs \texttt{No}  then there is no tuple $(S'_1, \ldots S'_r,R') \in \mathcal{U}^{r+1}$ such that $R \subseteq R'$, $\forall i \in [1,r]: S_i \subseteq S'_i \subseteq R'$, $(S'_i \setminus S_i) \cap R = \emptyset$,   $(I_i,S'_i) \vdash \Pi_i$ and $\texttt{dist}(S'_1, \ldots, S'_r) \geq d$.
  \end{claim}
  
  \begin{proofclaim}
    We again provide a proof of correctness using a maximal counter example.
    Let $(I = (I_1, \ldots, I_r),d, (S^*_1, \ldots, S^*_r,R^*))$ be a valid input of Algorithm~\ref{alg:subD}.
    Suppose that Algorithm~\ref{alg:subD} returns \texttt{No} and assume by contradiction that
    there exists a tuple $\hat{\mathcal{S}} = (\hat{S}_1, \ldots \hat{S}_r,\hat{R}) \subseteq \mathcal{U}^{r+1}$ such that $R^* \subseteq \hat{R}$, $\forall i \in [1,r]: S^*_i \subseteq \hat{S}_i \subseteq \hat{R}$, $(\hat{S}_i \setminus S^*_i) \cap R = \emptyset$,   $(I_i,\hat{S}_i) \vdash \Pi_i$ and $\texttt{dist}(\hat{S}_1, \ldots, \hat{S}_r) \geq d$.
    Over all the recursive calls, including the initial call, of Algorithm~\ref{alg:subD} that returns \texttt{No}, we consider the one where
    the part $(S_1, \ldots, S_r,R)$ of the input $(I,d, (S_1, \ldots, S_r,R))$ is maximal with regard to $\prec_{r+1}$ and such that
    there exists a tuple $(\hat{S}_1, \ldots \hat{S}_r,\hat{R}) \subseteq \mathcal{U}^{r+1}$ such that $R \subseteq \hat{R}$, $\forall i \in [1,r]: S_i \subseteq \hat{S}_i \subseteq \hat{R}$, $(\hat{S}_i \setminus S_i) \cap R = \emptyset$,   $(I_i,\hat{S}_i) \vdash \Pi_i$ and $\texttt{dist}(\hat{S}_1, \ldots, \hat{S}_r) \geq d$.

    Observe that Algorithm~\ref{alg:subD} returns \texttt{No} at line 4 or line 19.

    In the case where for each $i \in [1,r]$, $(I_i, S_i) \vdash \Pi_i$, we have that, by Lemma~\ref{rk:natdiv} and definition of $(\hat{S}_1, \ldots \hat{S}_r,\hat{R})$ that
    $d \leq \texttt{dist}(\hat{S}_1, \ldots, \hat{S}_r) \leq \texttt{dist}({S}_1, \ldots, {S}_r)$.
    Thus, the conditions of line 3 cannot be satisfied and Algorithm~\ref{alg:subD} cannot return \texttt{No} at line 4.
    So we can assume that Algorithm~\ref{alg:subD} returns \texttt{No} at line 19.
    
    Let $i_0$ the value in $[1,r]$ selected at line 6 and $X$ the value returned at line 7.
    As we have assumed that there exists $\hat{S}_{i_0}$ such that $ S_{i_0} \subseteq \hat{S}_{i_0}$ and $(I_{i_0},\hat{S}_{i_0}) \vdash \Pi_{i_0}$ and we have $(I_{i_0},{S}_{i_0}) \not \vdash \Pi_{i_0}$, we know that $X$ cannot be \texttt{No} or empty. Indeed, $(\hat{S}_{i_0}\sm S_{i_0})\cap R = \emptyset$, hence $\hat{S}_{i_0}\sm S_{i_0}$ is a valid output for $ \mathcal{A}_{i_0} ({I_{i_0}},  S_{i_0}, R)$.
    
    Now, for each $i \in [1,r]$, we define $X_i = X \cap (\hat{S}_i \setminus S_i)$.
    Then the tuple $(X_1, \ldots, X_r)$ is one of the tuples considered in the for loop of line 9.
    
    In the case where $\bigcup_{i \in [1,r]} X_i \neq \emptyset$, as $(S_1,S_2,\dots, S_r,R) \prec_{r+1} (S_1\cup X_1,\dots, S_r\cup X_r, R \cup X)$, we have by definition of our maximal counter example that line 12 returns a ``Result'' that is not \texttt{No} as $\hat{\mathcal{S}}$ is also a solution for the recursive call.
    
    In the case where $\bigcup_{i \in [1,r]} X_i = \emptyset$,
    we define the tuple $\hat{\mathcal{S}}' = (\hat{S}'_1, \ldots \hat{S}'_r,\hat{R}') \subseteq \mathcal{U}^{r+1}$,
    where $\hat{S}'_{i_0} = S_{i_0} \cup X$, $\hat{R}' = (\hat{R} \setminus \hat{S}_{i_0}) \cup \hat{S}'_{i_0}$  and for each $i \in [1,r] \setminus \{i_0\}: \hat{S}'_i = \hat{S}_i$.
    By definition of $\hat{\mathcal{S}}$ and Definition~\ref{def:natdiv}, we have that
    $d \leq \texttt{dist}(\hat{S}_1, \ldots, \hat{S}_r) \leq \texttt{dist}(\hat{S}'_1, \ldots, \hat{S}'_r)$.
    Moreover we have that $R \subseteq \hat{R}'$, $\forall i \in [1,r]: S_i \subseteq \hat{S}'_i \subseteq \hat{R}'$, $(\hat{S}'_i \setminus S_i) \cap R = \emptyset$,   $(I_i,\hat{S}'_i) \vdash \Pi_i$.
    This implies that the tuple $(X'_1, \ldots, X'_r)$, where $X'_{i_0} = X$ and for each $i \in [1,r] \setminus \{i_0\}: X'_i = \emptyset$, of the for loop of line 9, leads, by definition of our maximal counter example, to a ``Result'' at line 12 that is not \texttt{No}.
  \end{proofclaim}
  
  Combining {(3.1)} and {(3.2)} concludes the proof of the lemma.
\end{proof}

We can now discuss the running time of Algorithm 1.
\begin{lemma}
  \label{lemma:runtime}
  Assuming that the recursive calls can be done by an oracle in constant time, Algorithm~\ref{alg:subD} proceed in
  $O(r \cdot \rtprt{I}{\kmax} + \rtdist{r}{\kmax} + \max_{i\in[1,r]}\max_{A \subseteq B \subseteq \mathcal{U}}f_i(|I_i|,|A|,|B|) + 2^{r \cdot \kmax}\cdot (r + \kmax)^{O(1)}) $ steps
  where
  
  \begin{itemize}
  \item $k_i$ is the maximum size of a solution $S$ such that $(I_i,S) \vdash \Pi_i$ for each $i \in [1,r]$,
    $\ksum = \sum_{i \in [1,r]}k_i$ and $\kmax = \max_{i \in [1,r]}k_i$.
  \item  $\rtprt{I}{q}$ denotes the maximum runtime for testing $(I_i,S) \vdash \Pi_i$ over every $i \in [1,r]$ and $(I_i,S) \in \inputset \times \solutionset$ such that  $|S| \leq q$. 
  \item $\rtdist{p}{q}$ denotes the maximum runtime for computing $\texttt{dist}(S_1, \ldots, S_p)$ where for each $i \in [1,p]$, we have $|S_i| \leq q$.
  \end{itemize}
\end{lemma}

\begin{proof}
  As both Line 1 and Line 3 consist of at most $r$ properties testing and one distance computing, each of them can be computed in time $r \cdot \rtprt{I}{\kmax} + \rtdist{r}{\kmax}$.
  Line 7 is a call of \textsc{$\Pi_{i_0}$-Completion}, it is done in time at most $\max_{i\in[1,r]}\max_{A \subseteq B \subseteq \mathcal{U}}f_i(|I_i|,|A|,|B|)$.
  Then in Line 9, we have our main loop.
  As $X$ is of size at most $\kmax$, we have at most $2^{\kmax}$ different subsets of $X$ and so our loop has at most $2^{r \cdot \kmax}$ iterations.
  Moreover, inside the loop, everything can be done in polynomial time in $r+\kmax$.
  This concludes the proof.
\end{proof}

We now have all we need to prove Theorem~\ref{th:main}.

\main*

\begin{proof}[\emph{Theorem~\ref{th:main}.}]
  Let $(I=(I_1, \ldots, I_r),d)$ be an input of \textsc{\texttt{dist}-Diverse-($\Pi_1$,$\ldots$, $\Pi_r$)}.
  We apply Algorithm~\ref{alg:subD} on the input $(I,d,(\emptyset, \ldots, \emptyset, \emptyset))$.
  By construction, $(I,d,(\emptyset, \ldots, \emptyset, \emptyset))$ is a valid input of Algorithm~\ref{alg:subD} and so, by Lemma~\ref{lemma:correctness}, Algorithm~\ref{alg:subD} correctly returns a solution of \textsc{\texttt{dist}-Diverse-($\Pi_1$,$\ldots$, $\Pi_r$)}.
  
  Let us now consider the running time of the algorithm.
  This algorithm is a branching algorithm. The depth of the branching tree is at most $\ksum = \sum_{i \in [1,r]}k_i$ as in each recursive call, at least one element is added to the current solution sets.
  Moreover, at each step, the algorithm branches at most $2^{r \cdot \kmax}$ times, one for each $(X_1,\ldots, X_r) \subseteq X^r$.
  Thus the number of nodes in the recursive tree is  $O(2^{r \cdot \kmax \cdot \ksum})$.
  Combining this with Lemma~\ref{lemma:runtime}, we obtain the claimed running time.  
\end{proof}

\section{Compliant distances}
\label{sec:dist}

In this section, we explicitly show how the most common distances between sets can be used in Theorem~\ref{th:main} by pointing out some situations in which they can be associated to a property to form a naturally diverse pair. 

\subsection{Disjoint distance}

Given an integer $r$, a set $\elementset$, and $\solutionset = 2^{\elementset}$, we define $\texttt{r-disjoint}: \solutionset^r \to \{0,1\}$ such that for each $(S_1, \ldots, S_r) \in \solutionset^r$ we have:
\begin{itemize}
\item $\texttt{r-disjoint}(S_1, \ldots, S_r) = 1$ if for each $1 \leq i < j \leq r$, $S_i \cap S_j = \emptyset$
\item $\texttt{r-disjoint}(S_1, \ldots, S_r) = 0$ otherwise.
\end{itemize}

\begin{lemma}
  \label{lemma:disjoint}
  Let $r$ be an integer and $\Pi$ be a property. $(\Pi,\texttt{r-disjoint})$ is naturally diverse.
\end{lemma}

\begin{proof}
  Let $I \subseteq \inputset_{\Pi}$,
  let $(S_1,\ldots, S_r, X, Y) \in {\solutionset_{\Pi}}^{r+2}$
  such that $S_j \cap Y = \emptyset$ for each $j \in [1,r]$, and
  let $i \in [1,r]$
  such that $X \subseteq S_i$, $(I,S_i) \vdash \Pi$, and $(I,(S_i \setminus X) \cup Y) \vdash \Pi$.
  Without loss of generality, we assume that $i = 1$.
  If $\texttt{r-disjoint}(S_1, \ldots, S_r) = 0$, we have that $\texttt{r-disjoint}((S_1 \setminus X) \cup Y, \ldots, S_r)$, that can take only two values $0$ and $1$, has to be greater or equal to $\texttt{r-disjoint}(S_1, \ldots, S_r)$.
  If $\texttt{r-disjoint}(S_1, \ldots, S_r) = 1$, then $S_1, \ldots, S_r$ are pairwise disjoint and so, as $Y$ is also disjoint to every $S_j$, $j \in [1,r]$, we have that $\texttt{r-disjoint}((S_1 \setminus X) \cup Y, \ldots, S_r) = 1$.
  This concludes the proof.
\end{proof}

\subsection{Distances working with fixed size properties}
\label{sec:fospdistance}

In this section, we start by providing a general statement, and then we show how to apply it for specific distances.

\subsubsection{From 2-distance functions to r-distance functions}

In general, the distances are defined between two sets. In order to fit for $r$ sets, we define the sum and the minimum variants for each such distance. More formally, if $\texttt{Dist}: \mathcal{U}^2 \to \mathbb{R}$ is a distance, we define
\begin{equation}
  r\mbox{-}\texttt{SumDist}(X_1, \ldots, X_r) = \sum_{1 \leq i < j \leq r}\texttt{Dist}(X_i,X_j)
\end{equation}
and
\begin{equation}
  r\mbox{-}\texttt{MinDist}(X_1, \ldots, X_r) = \min_{1 \leq i < j \leq r}\texttt{Dist}(X_i,X_j)
\end{equation}

\begin{lemma}
  \label{lemma:distgeneric}
  Let $r$ be an integer and $\Pi$ be a fixed size property.
  Let $\texttt{Dist}: \mathcal{U}_{\Pi}^2 \to \mathbb{R}$ be a distance,
  such that for each $S_1,S_2, X, Y \subseteq \mathcal{U}_{\Pi}$ verifying $X \subseteq S_1$, $Y \cap S_1 = Y \cap S_2 = \emptyset$, and $|X| = |Y|$ we have
  $\texttt{Dist}(S_1,S_2) \leq \texttt{Dist}((S_1\setminus X) \cup Y,S_2)$.
  Then  $(\Pi,\texttt{r-MinDist})$, and $(\Pi,\texttt{r-SumDist})$ are naturally diverse.
\end{lemma}

\begin{proof}
  Let $r$, $\Pi$, and $\texttt{Dist}$ be as in the lemma's statement.
  Let $I \subseteq \inputset_{\Pi}$,
  $(S_1,\ldots, S_r, X, Y) \subseteq {\mathcal{U}_{\Pi}}^{r+2}$
  such that, for each $j \in [1,r]$, $S_j \cap Y = \emptyset$, and
  let $i \in [1,r]$
  such that $X \subseteq S_i$, $(I,S_i) \vdash \Pi$, and $(I,(S_i \setminus X) \cup Y) \vdash \Pi$.
  Note that, as $\Pi$ is a fixed size property, $|X| = |Y|$.
  For readability, we assume that $i = 1$, possibly by swapping the sets $S_1$ and $S_i$.
  
  For the \texttt{r-SumDist} distance.
  We have that
  \begin{align*}
    \texttt{r-SumDist}(S_1, \ldots, S_r)&=\sum_{1 < j \leq r}\texttt{Dist}(S_1,S_j) + \sum_{2 \leq i < j \leq r} \texttt{Dist}(S_i,S_j)\\
                                        &\leq \sum_{1 < j \leq r} \texttt{Dist}((S_1\setminus X) \cup Y,S_j) + \sum_{2 \leq i < j \leq r} \texttt{Dist}(S_i,S_j)\\
                                        &\leq\texttt{r-SumDist}((S_1 \setminus X) \cup Y, \ldots, S_r).
  \end{align*}
  
  For the \texttt{r-MinDist} distance.
  We have that
  \begin{align*}
    \texttt{r-MinDist}(S_1, \ldots, S_r)&=\min\big(\min_{1 < j \leq r}\texttt{Dist}(S_1,S_j) , \min_{2 \leq i < j \leq r} \texttt{Dist}(S_i,S_j)\big)\\
                                        &\leq \min\big(\min_{1 < j \leq r} \texttt{Dist}(S_1\setminus X \cup Y,S_j), \min_{2 \leq i < j \leq r} \texttt{Dist}(S_i,S_j)\big)\\
                                        &\leq\texttt{r-MinDist}((S_1 \setminus X) \cup Y, \ldots, S_r).
  \end{align*}
  This concludes the proof.
\end{proof}

\subsubsection{Application to classical distance functions}

In this section, we consider several $2$-distance functions that are often used when considering diversity.
Namely, we consider the hamming distance, the Jaccard distance and the Otsuka–Ochiai coefficient.

\paragraph{Hamming distance.}
The hamming distance is one of the most classical distances to measure the changes between two sets.
Formally, given a set $\elementset$ and $\solutionset = 2^{\elementset}$, define $\texttt{Ham}: \mathcal{U}^2 \to \mathbb{N}$ as
\begin{equation}
  \texttt{Ham}(X_1,X_2) = |X_1\setminus X_2| +  |X_2\setminus X_1| 
  \label{Hamdistance}
\end{equation}

\paragraph{Jaccard distance.}

The Jaccard distance is a common tool in statistics that evaluates the dissimilarity of sample sets.
Formally, given a set $\elementset$ and $\solutionset = 2^{\elementset}$, we define $\texttt{Jac}: \mathcal{U}^2 \to \mathbb{R}$ as
\begin{equation}
  \texttt{Jac}(X_1,X_2) = \frac{\vert X_1 \cup X_2 \vert - \vert X_1 \cap X_2 \vert}{\vert X_1 \cup X_2 \vert}
  \label{jaccdistance}
\end{equation} 

\paragraph{Otsuka–Ochiai coefficient.}
The Otsuka–Ochiai coefficient is a particular case of the cosine similarity where the vector can only take binary values.
Formally, given a set $\elementset$ and $\solutionset = 2^{\elementset}$, we define $\texttt{Ots}: \mathcal{U}^2 \to \mathbb{R}$ as
\begin{equation}
  \texttt{Ots}(X_1,X_2) = 1 - \frac{| X_1 \cap X_2 |}{\sqrt{|X_1| \cdot |X_2|}}
  \label{Ots}
\end{equation}

\begin{lemma}
  \label{lemma:dist}
  Let $\Pi$ be a fixed size property. The pairs:
  \begin{itemize}
  \item $(\Pi,\texttt{r-MinHam})$, 
    $(\Pi,\texttt{r-SumHam})$,
  \item $(\Pi,\texttt{r-MinJac})$, 
    $(\Pi,\texttt{r-SumJac})$,
  \item $(\Pi,\texttt{r-MinOts})$, and 
    $(\Pi,\texttt{r-SumOts})$
  \end{itemize}
  are naturally diverse.
  
\end{lemma}

\begin{proof}
  Let $S_1,S_2, X, Y \subseteq \mathcal{U}_{\Pi}$ be such that $X \subseteq S_1$, $Y \cap S_1 = Y \cap S_2 = \emptyset$, and $|X| = |Y|$.
  Observe that, since $|X|=|Y|$ and $X\subseteq S_1$,  we have that $ |((S_1\sm X)\cup Y)| =|S_1|$. 
  In addition    $ |S_2\cap ((S_1\sm X)\cup Y)| =  |S_2\cap (S_1\sm X)| + |S_2\cap Y|  = |S_2\cap S_1| - |S_2\cap X|$ as $ Y\cap S_2=\emptyset$.
  Hence :

  $ \begin{array}{rcl}
    \texttt{Ham}((S_1\setminus X) \cup Y,S_2)	 &=& |S_1| + |S_2|  - 2(  |S_2\cap S_1| - |S_2\cap X| )\\
                                                 &=& |S_1|+ |S_2| -  2( |S_2\cap S_1| ) +2|S_2\cap X| \\
                                                 &=& |S_1\sm S_2|+ |S_2\sm S_1| +2|S_2\cap X| \\
                                                 &=&   \texttt{Ham}(S_1,S_2) +2|S_2\cap X|\\
                                                 &\geq&  \texttt{Ham}(S_1,S_2)
  \end{array}$

  $ \begin{array}{rcl}
    \texttt{Jac}((S_1\setminus X) \cup Y,S_2)	 &=& 1 - \frac{|((S_1\sm X)\cup Y)\cap S_2|}{|((S_1\sm X)\cup Y)\cup S_2|}\\
                                                 &=& 1 - \frac{|S_1\cap S_2| - |S_2\cap X|}{|((S_1\sm X)\cup Y)|+|S_2| - |((S_1\sm X)\cup Y)\cap S_2|}\\\
                                                 &=& 1 - \frac{|S_1\cap S_2| - |S_2\cap X|}{|S_1|+|S_2| -|S_1\cap S_2| + |S_2\cap X|}\\\
                                                 &=&   1 - \frac{|S_1\cap S_2| - |S_2\cap X|}{|S_1\cup S_2| +  |S_2\cap X|}\\\
                                                 &\geq&   1 - \frac{|S_1\cap S_2|}{|S_1\cup S_2| }\\\
                                                 &\geq& 	 \texttt{Jac}(S_1,S_2)
  \end{array}$
  
  $ \begin{array}{rcl}
    \texttt{Ots}((S_1\setminus X) \cup Y,S_2)	 &=& 1 - \frac{|((S_1\sm X)\cup Y)\cap S_2|}{\sqrt{|(S_1\sm X)\cup Y)|\cdot |S_2|}}\\
                                                 &=& 1 - \frac{|S_1\cap S_2| - |S_2\cap X|}{\sqrt{|S_1|\cdot |S_2|}}\\
                                                 &\geq&   1 - \frac{|S_1\cap S_2|}{\sqrt{|S_1|\cdot |S_2|}}\\
                                                 &\geq& 	\texttt{Ots}(S_1,S_2)	 
  \end{array}$
  
  Applying Lemma~\ref{lemma:distgeneric} concludes the proof.
\end{proof}

\section{Application exemples}
\label{sec:ex}

In this section we show how to apply our meta-algorithm to two problems, namely \textsc{Vertex Cover} for which the solution space is a set of vertices and \textsc{Steiner Tree} for which the solution space is a set of edges.

\subsection{Minimum Vertex Cover}

We remind the reader that a vertex cover of a graph $G$ is a set $S \subseteq V(G)$ such that for each $e \in E(G)$, $S \cap e = \emptyset$.
In this section, we focus on the property $\Pi_{mvc}$ that is defined as $(G,S) \in \mathcal{G} \times \mathcal{U}$ satisfy $\Pi_{mvc}$ if and only if $S$ is a vertex cover of $G$ of minimum size, where $\mathcal{U} = 2^{\mathcal{E}}$ and $\mathcal{E}$ is the set of all vertices of graphs in $\mathcal{G}$.

We can first notice that, as we require $S$ to be of minimum size, we have that $\Pi_{mvc}$ is a fixed size property.
This implies that, using Lemma~\ref{lemma:disjoint} and \ref{lemma:dist}, we have that
\begin{lemma}
  \label{lemma:mvc}
  $(\Pi_{mvc}, \texttt{dist})$ is naturally diverse for $\texttt{dist}$ being 
  $\texttt{r-disjoint}$,
  $\texttt{r-MinHam}$, 
  $\texttt{r-SumHam}$,
  $\texttt{r-MinJac}$, 
  $\texttt{r-SumJac}$,
  $\texttt{r-MinOts}$, or
  $\texttt{r-SumOts}$.
\end{lemma}
Moreover, in order to apply Theorem~\ref{th:main}, we also need to show that there exists an algorithm that solves \textsc{$\Pi_{mvc}$-Completion}. We first restate the completion problem in the context of $\Pi_{mvc}$.

\probl{$\Pi_{mvc}$-completion}
{A graph $G \in \mathcal{G}$ and two vertex sets $A$ and $B$ in $V(G)$ such that $A \subseteq B$.}
{A set $S \in V(G) \setminus B$ such that $(G,A \cup S) \vdash \Pi_{mvc}$  or a correct statement that such a set does not exist.}

The following lemma is a folklore result that can be proven, for instance, using bounded search trees as described in~\cite[Section 3.1]{CyFoKoLoMaPiPiSa2015}.
\begin{lemma}
  \label{lemma:mvcc}
  \textsc{$\Pi_{mvc}$-Completion} can be solved in time $2^{O(k)}\cdot n^{O(1)}$ where $k$ is the size of a minimum vertex cover of the input graph.
\end{lemma}

Combining Theorem~\ref{th:main} with Lemma~\ref{lemma:mvc} and Lemma~\ref{lemma:mvcc}, we obtain the following Theorem.
\begin{theorem}
  The \textsc{\texttt{dist}-Diverse-($\Pi_1$,$\ldots$, $\Pi_r$)} problem where
  \begin{itemize}
  \item $\Pi_1= \ldots = \Pi_r = \Pi_{mvc}$ and 
  \item \texttt{dist} is one of \texttt{disjoint}, \texttt{SumHam}, \texttt{MinHam}, \texttt{SumJac}, \texttt{MinJac}, \texttt{SumOts}, and \texttt{MinOts}
  \end{itemize}
  can be solved in time $(2^{O(k)}\cdot n^{O(1)} + 2^{r \cdot k}) \cdot 2^{r^2 \cdot k^2} \cdot n^{O(1)}$ on an instance $(I = (G,\ldots, G),d)$ where $k$ is the size of a minimum vertex cover of the input graph $G$.
\end{theorem}

\paragraph{Discussion.}
We presented \textsc{Minimum Vertex Cover} as it is a classical problem, but we want to mention that for distance with strong properties, namely \texttt{disjoint}, \texttt{SumHam}, and \texttt{MinHam}, better algorithms are known.
For the \texttt{disjoint}, it is easy to see that if $r \geq 3$, no solution exists and for $r = 2$ the graph must be bipartite in order to have a solution.
For \texttt{SumHam} and \texttt{MinHam}, the results from \cite{fvBaFeJaMaOlPhRo2022} and \cite{DrMa2024} provide algorithms with running time of order $2^{r \cdot k} \cdot n^{O(1)}$.

However, the techniques used for \texttt{disjoint}, \texttt{SumHam}, and \texttt{MinHam} cannot be naturally extended to other distances.
Thus, to the best of our knowledge, we provide the first results for the distance \texttt{SumJac}, \texttt{MinJac}, \texttt{SumOts}, and \texttt{MinOts}.

\subsection{Minimum Steiner Tree}

In this section, we focus on Minimum Steiner Tree. The property $\Pi_{mst}$  is defined as $((G,T),S) \in (\mathcal{G} \times \mathcal{V}) \times \mathcal{U}$ satisfy $\Pi_{mst}$ if and only if  $S \subseteq E(G)$ is an edge set of minimum size such that $(V(S),S)$ is a connected subgraph of $G$ and such that $T \subseteq V(S)$ , where $\mathcal{V}$ is the set of all subsets of vertices of graphs in $\mathcal{G}$, $\mathcal{U} = 2^{\mathcal{E}}$ and $\mathcal{E}$ is the set of all edges of graphs in $\mathcal{G}$.

\probl{Minimum Steiner Tree}
{A graph $G$, a set $T \subseteq V(G)$.}
{A set $S \subseteq E(G)$ such that $((G,T),S) \vdash \Pi_{mst}$ or a correct statement that such a set does not exist.}

As in the previous section, as we require $S$ to be of minimum size, we have that $\Pi_{mst}$ is a fixed size property.
This implies that, using Lemma~\ref{lemma:disjoint} and \ref{lemma:dist}, we have that
\begin{lemma}
  \label{lemma:mst}
  $(\Pi_{mst}, \texttt{dist})$ is naturally diverse for $\texttt{dist}$ being 
  $\texttt{r-disjoint}$,
  $\texttt{r-MinHam}$, 
  $\texttt{r-SumHam}$,
  $\texttt{r-MinJac}$, 
  $\texttt{r-SumJac}$,
  $\texttt{r-MinOts}$, or
  $\texttt{r-SumOts}$.
\end{lemma}

Again, in order to apply Theorem~\ref{th:main}, we also need to show that there exists an algorithm that solves \textsc{$\Pi_{mst}$-Completion}. We first restate the completion problem in the context of $\Pi_{mst}$.

\probl{$\Pi_{mst}$-completion}
{A graph $G \in \mathcal{G}$, a set $T \subseteq V(G)$ and two edge sets $A$ and $B$ in $E(G)$ such that $A \subseteq B$.}
{A set $S \subseteq E(G) \setminus B$ such that $((G,T),A \cup S) \vdash \Pi_{mst}$  or a correct statement that such a set does not exist.}

It is known that \textsc{Minimum Steiner Tree} can be solved in \texttt{FPT} time when parameterized by the size of $T$. The current best known algorithm is given in~\cite{FuKeMoRiRoWa2007} by the following theorem.

\begin{theorem}[\cite{FuKeMoRiRoWa2007}]
  The \textsc{Minimum Steiner Tree} problem can be solved in time $c^t\cdot n^{O(1)}$, where $t = |T|$ is the number of terminals, for any $c > 2$.
\end{theorem}

Since any solution of \textsc{Minimum Steiner Tree} with $p$ edges is such that $t \leq p + 1$, we deduce the following lemma.

\begin{lemma}
  \label{co:steiner}
  The \textsc{Minimum Steiner Tree} problem, on an instance $(G,T)$, can be solved in time $c^k\cdot n^{O(1)}$ for any $c > 2$ where $k$ is the size of a minimum size solution.
\end{lemma}

This again provides an algorithm for the \textsc{$\Pi_{mst}$-Completion}.
\begin{lemma}
  \label{co:steinercomp}
  The \textsc{$\Pi_{mst}$-Completion} problem can be solved in time $c^k\cdot n^{O(1)}$ for any $c > 2$ where $k$ is the size of a minimum size solution.
\end{lemma}
\begin{proof}
  Let $I = (G,T,A,B)$ be an input of \textsc{$\Pi_{mst}$-Completion}.
  We present two operations on $I$ and show that applying them creates an equivalent instance of the problem. 
  We consider that two instances are equivalent if it is possible to obtain in linear time a solution to one from a solution to the other.

  \textbf{Reduction 1.} If $B \setminus A \not = \emptyset$, we select $e \in B \setminus A$ and construct $I' = (G' = (V(G), E(G) \setminus \{e\}), T, A, B \setminus \{e\})$.
  \vspace{2ex}
  
  \textit{Correctness of Reduction 1.} As a solution of \textsc{$\Pi_{mst}$-Completion} on $I$ cannot contain an element of $B \setminus A$, we have that any solution $S$ of \textsc{$\Pi_{mst}$-Completion} on $I$ is a solution of \textsc{$\Pi_{mst}$-Completion} on $I'$ and vice versa.
  
  \vspace{4ex}
  
  \textbf{Reduction 2.} If $A \not = \emptyset$ and $B \setminus A = \emptyset$, let $e = \{u,v\} \in A$ and let $w$ be a fresh vertex. Let $V' =  (V(G) \setminus \{u,v\}) \cup \{w\}$.

  For $i \in \{0,1,2\}$, let $E_i = \{z \in E(G) \mid | z \cap e | = i\}$. Observe that $E_2=\{e\}$ and so $E_0\cup E_1=E\setminus \{e\}$. Let $f: E_0 \cup E_1 \to V' \times V'$ be such that for each $z \in E_0$, $f(z) = z$ and for each $z \in E_1$, $f(z) = z \setminus e \cup \{w\}$.  For $X\subseteq E$, set $f(X)= \{f(z) \mid z\in X\cap (E_0\cup E_1\}$.Let $E' = \{f(z) \mid z \in E(G) \setminus \{e\}\}$.
  
  Let $I' = (G',T',A',B')$ be such that
  \begin{itemize}
  \item $G' = (V',E')$
  \item $T' = (T \setminus \{u,v\}) \cup \{w\}$,
  \item $A' = \{ f(z) \mid z \in (E_0 \cup E_1) \cap A\}$ (ie, $A'=f(A\sm \{e\})$), 
  \item $B' = \{ f(z) \mid z \in (E_0 \cup E_1) \cap B\}$ (ie, $B'=f(B\sm \{e\})$),
  \end{itemize}

  Roughly speaking, $G'$ is obtained from $G$ by contracting the edge $e$ and the new vertex created after the contraction is added as terminal.
  $f$ map the old edges in $G$ to the new one in $G'$.
  \vspace{2ex}

  \textit{Correctness of Reduction 2.}
  
  For the sake of the proof, we need the additional following definition.
  Let $g: E' \to E(G) \setminus \{e\}$ such that for each $z' \in E'$, 
  \begin{itemize}
  \item if $z'\cap \{w\}=\emptyset$, $g(z')=z$,
  \item if $z'=\{x,w\}$ and $\{x,u\} \in E$, $g(z')= \{x,u\}$. 
  \item if $z'= \{x,w\}$ and $\{x,u\} \notin E$ , $g(z')= \{v,u\}$. 
  \end{itemize} 
  Roughly speaking, $g$ inverses the function $f$ that is not injective. Indeed, if $\{x,u\} , \{x,v\} \in E$ then $f(\{x,u\})=f(\{x,v\})=\{w,x\}$. In that case we arbitrarily choose $g(\{w,x\})=\{u,x\}$.

  We now expose some remarks that will be used through the proof. 
  Given $X\subseteq E$, the following properties are satisfied.
  \begin{enumerate}
  \item $V(X)\setminus \{u,v\} = V(f(X))\setminus \{w\}$\label{rmdc:ensemble_sommets}
  \item $f(X)=X\setminus (E_1\cup \{\{u,v \}\}) \cup f(E_1\cap X)$. \label{rmdc:ensembles_arretes}
  \item $|f(X)|=|X|- |\{ x \in V(X) \mid \{u,x\} \in X \mbox{ and } \{v,x\}\in X \}|$. Indeed, for $x\notin \{u,v\}$, if $\{u,x\},\{v,x\}\in X$, $f(\{u,x\})=f(\{v,x\})=\{w,x\}$. \label{rmdc:tailleEnsembleParf}
  \item For $X'\subseteq E'$, $|g(X')|=|X'|$.  Indeed, $g$ is injective. \label{rmdc:tailleEnsembleParg}
  \end{enumerate}
  
  We now start by proving two claims stating that there is a solution on $G$ possibly not minimal if and only if there is a solution on $G'$ possibly not minimal. Then we prove that if one is minimal, so is the other.

  \begin{claim}\label{claim:spannintreeSDonneS'}
    If $S\subseteq E\setminus B$ is a set such that  $A\cup S \subseteq E$, $(V(A\cup S),A\cup S)$ is connected and $T\subseteq V(A\cup S)$  then  $S'=f(S)\subseteq E'\setminus B'$ is such that  $A'\cup S' \subseteq E'$, $(V(A'\cup S'),A'\cup S')$ is connected and $T'\subseteq V(A'\cup S')$. 
  \end{claim}
  
  \begin{proofclaim}
    
    Observe first that since $T\subseteq V(A\cup S)$, by remark \ref{rmdc:ensemble_sommets}, $T'\subseteq V(A'\cup S')$. 
    
    \vspace{2ex}
    
    Let $z'\in S'$, by definition, there exists $z\in S\cap (E_0\cup E_1)$ such that $f(z)=z'$. Since $S\subseteq E\setminus B$, then $z\in (E\setminus B)\cap (E_0\cup E_1)$. Hence $z'=f(z)\in E'\setminus B'$ and so $S'\subseteq E'\setminus B'$. In addition, since $f(A)\subseteq f(E)$, $A'\cup S'\subseteq E'$. 
    
    \vspace{2ex}
    
    We still have to prove that $(V(A'\cup S'),A'\cup S')$ is connected. By remark \ref{rmdc:ensemble_sommets}, $V(A\cup S)\setminus \{u,v\} = V(A'\cup S')\setminus \{w\}$.  Consider two vertices $x,y\in V(A'\cup S')$. For this part of the proof, an illustration is provided in Figure~\ref{f:sub}.

    First suppose $w \notin \{x,y\}$ and so by remark~\ref{rmdc:ensemble_sommets}, $\{x,y\}\in A\cup S$. Since $(V(A\cup S),A\cup S)$ is connected, there exists a $\{x,y\}$-path $p$ in $G[A\cup S]$.  If $V(p)\cap \{u,v\}=\emptyset$ then $p\subseteq A'\cup S'$. Otherwise, there exists two vertices $u',v'$ such that either  $\{\{u',u\}, \{u,v\}, \{v,v'\}\}\subseteq p$, $\{\{u',u\}, \{u,v'\}\} \subseteq p$, or $\{\{u',v\}, \{v,v'\}\}\subseteq p$. In all cases,  $\{\{u',w\}, \{w,v'\}\}\subseteq A'\cup S'$ as, if they exists, $\{\{u,v\}, \{u',u\},\{v,v'\} ,\{u',v\}, \{u,v'\}\}\subseteq A\cup S$. Hence, $p\setminus\{\{u',u\},\{u,v\},\{v,v'\},\{u,v'\},\{u',v\}\}\cup \{\{u',w\}, \{w,v'\}\}$ is a $\{x,y\}$-path in $G'[A'\cup S']$. In all cases, if $w \notin \{x,y\}$, there exists a $\{x,y\}$-path in $G'[A'\cup S']$. 
    
    Without loss of generality assume that $x=w$. Since $(V(A\cup S),A\cup S)$ is connected and since $\{u,v\}\in A$, there exists $w'\in V(A\cup S)$ such that either $\{w',u\}\in A\cup S$ and $f(\{w',u\})=\{w',w\}$ or $\{w',v\}\in A\cup S$ and $f(\{w',v\})=\{w',w\}$. In both cases, $\{w,w'\}\in A'\cup S'$ and by previous paragraph, there exists a $\{y,w'\}$-path in $G'[A'\cup S']$ that, adding the edge $\{w',w\}$ yields a $\{x,y\}$-path in $G'[A'\cup S']$. Hence $(V(A'\cup S'),A'\cup S')$ is connected. 
  \end{proofclaim}
  
  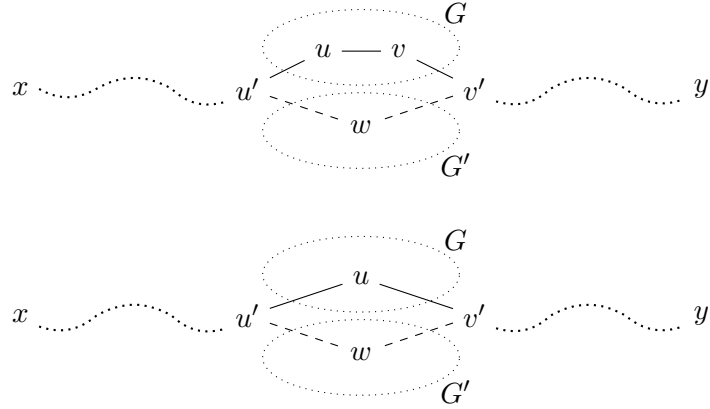
\begin{figure}
    \centering
    \begin{tikzpicture}
      
      \begin{scope}
        \node[-] (x) at (0,0) {$x$};
        \node[-](y) at (9,0) {$y$};
        \node[-] (u') at (3,0) {$u'$};
        \node[-] (u) at (4,0.5) {$u$};
        \node[-] (v) at (5,0.5) {$v$};
        \node[-] (v') at (6,0) {$v'$};
        \node[-] (w) at (4.5,-0.5) {$w$};
        \draw[dotted,thick] (x.east)  to[bend right=30] (1,0)to[bend left=30](2,0)to[bend right=30](u'); 
        \draw[-] (u') -- (u);
        \draw[-] (u) -- (v); 
        \draw[-] (v) -- (v'); 
        \draw[dotted,thick] (v') to[bend right=30] (7,0)to[bend left=30](8,0)to[bend right=30] (y); 
        \draw[dashed] (u') -- (w); 
        \draw[dashed] (v') -- (w); 
        
        \node[-] (G) at (5.75,1) {$G$};
        \draw[dotted] (4.5,0.55) ellipse (1.3cm and 0.5cm);
        
        \node[-] (G2) at (5.75,-1) {$G'$};
        \draw[dotted] (4.5,-0.55) ellipse (1.3cm and 0.5cm);
      \end{scope}

      \begin{scope}[yshift=-3cm]
        \node[-] (x) at (0,0) {$x$};
        \node[-](y) at (9,0) {$y$};
        \node[-] (u') at (3,0) {$u'$};
        \node[-] (u) at (4.5,0.5) {$u$};
        \node[-] (v') at (6,0) {$v'$};
        \node[-] (w) at (4.5,-0.5) {$w$};
        \draw[dotted,thick] (x) to[bend right=30] (1,0)to[bend left=30](2,0)to[bend right=30](u'); 
        \draw[-] (u') -- (u);
        \draw[-] (u) -- (v'); 
        \draw[dotted,thick] (v') to[bend right=30] (7,0)to[bend left=30](8,0)to[bend right=30] (y); 
        \draw[dashed] (u') -- (w); 
        \draw[dashed] (v') -- (w);
        
        \node[-] (G) at (5.75,1) {$G$};
        \draw[dotted] (4.5,0.55) ellipse (1.3cm and 0.5cm);
        
        \node[-] (G2) at (5.75,-1) {$G'$};
        \draw[dotted] (4.5,-0.55) ellipse (1.3cm and 0.5cm);
      \end{scope}

    \end{tikzpicture}	
    \caption{Relation between a $\{x,y\}$-path of $G$ and a  $\{x,y\}$-path of $G'$. The top part illustrates the case where a $\{x,y\}$-path in $G$ contains the edge $\{u,v\}$. The bottom part illustrates the case where a $\{x,y\}$-path in $G$ intersects only $u$ (the case for $v$ is analogous). }\label{f:sub}
  \end{figure}
  
  \begin{claim}\label{claim:spannintreeS'DonneS}
    If $S'\subseteq E'\setminus B'$ is such that  $A'\cup S' \subseteq E'$, $(V(A'\cup S'),A'\cup S')$ is connected and $T'\subseteq V(A'\cup S')$ then  $S=g(S')\subseteq E\setminus B$ is such that  $A\cup S \subseteq E$, $(V(A\cup S),A\cup S)$ is connected and $T\subseteq V(A\cup S)$.
  \end{claim}
  
  \begin{proofclaim}
    
    Observe first that since $T\subseteq V(A\cup S)$, by remark \ref{rmdc:ensemble_sommets}, $T'\subseteq V(A'\cup S')$.
    
    \vspace{2ex}
    
    Let $z\in S$. By definition of $S$ there exists $z'\in S'$ such that $g(z')=z$. Since $S'\subseteq E'\setminus B'$, $z'\in E'\setminus B'$. Therefore $z\in E\setminus B$ and so $S\subseteq E\setminus B$. In addition, since $A\subseteq E$, $A\cup S\subseteq E$.

    \vspace{2ex}
    
    We still have to prove that   $(V(A\cup S),A\cup S)$ is connected. For this part of the proof, an illustration is provided in Figure~\ref{f:sub}.  Consider two vertices $x,y\in V(A\cup S)$.

    First suppose $\{x,y\} \cap \{u,v\} = \emptyset$ and so, by remark~\ref{rmdc:ensemble_sommets}, $\{x,y\}\in A'\cup S'$. Since $(V(A'\cup S'),A'\cup S')$ is connected, there exists a $\{x,y\}$-path $p$ in $G'[A'\cup S']$.  If $V(p)\cap \{w\}=\emptyset$ then $p\subseteq A\cup S$. Otherwise, there exits two vertices $u',v'$ such that $\{u',w\},\{w,v'\}\in p$ and so  $\{u',w\},\{w,v'\}\in A'\cup S'$. Now, either $g(\{w,v'\})=\{u,v'\}$ or $g(\{w,v'\})=\{v,v'\}$. In both cases $g(\{u',w\})\in A\cup S$. If $g(\{w,v'\})=\{u,v'\}$ then $p\setminus\{\{u',w\},\{w,v'\}\}\cup \{\{u',v\},\{v,v'\}\}$ is a $\{x,y\}$-path in $G[A\cup S]$. If $g(\{w,v'\})=\{v,v'\}$, recal that $\{u,v\}\in A$ and so $p\setminus \{\{u',w\},\{w,v'\}\}\cup \{\{u',u\},\{u,v\},\{v,v'\}\}$ is a $\{x,y\}$-path in $G[A\cup S]$. Hence, if $\{x,y\}\cap \{u,v\} = \emptyset$, there exists a $\{x,y\}$-path in $G'[A'\cup S']$. 
    
    Since, $\{u,v\}\in A$, assume, without loss of generality, that $x\in \{u,v\}$ and $y\notin \{u,v\}$. Since $(V(A'\cup S'),A'\cup S')$ is connected there exists $w'\in V(A'\cup S')$ such that  $\{w',w\}\in A'\cup S'$ and either $g(\{w,w'\})=\{u,w'\}$ or  $g(\{w,w'\})=\{v,w'\}$. In both cases  $g(\{w,w'\})\in A$.  By previous paragraph, there exists a $\{y,w'\}$-path  $p$ in $G'[A'\cup S']$. If $x=u$ then either $p\cup \{\{u,w'\}\}$  (if $g(\{w,w'\})=\{u,w'\}$) or  $p\cup \{\{u,v\},\{v,w'\}\}$  (if $g(\{w,w'\})=\{v,w'\}$) is a $\{x,y\}$-path in $G[A\cup S]$. If $x=v$ then either $p\cup \{\{u,w'\},\{u,v\}\}$  (if $g(\{w,w'\})=\{u,w'\}$) or  $p\cup \{\{v,w'\}\}$  (if $g(\{w,w'\})=\{v,w'\}$) is a $\{x,y\}$-path in $G[A\cup S]$. Hence $(V(A\cup S),A\cup S)$ is connected. \end{proofclaim}

  \begin{claim}
    If $S$ is a solution on $(G,T,A,B)$  then $f(S)$ is a solution on $(G',T',A',B')$.
  \end{claim}
  
  \begin{proofclaim}
    
    Suppose that $S$ is a solution on $(G,T,A,B)$ . By Claim \ref{claim:spannintreeSDonneS'},   $S'=f(S)\subseteq E'\setminus B'$ is such that  $A'\cup S' \subseteq E'$, $(V(A'\cup S'),A'\cup S')$ is connected and $T'\subseteq V(A'\cup S')$. We still have to prove that $S'$ is of minimum size. 
    
    By contradiction, let $R'$ be solution on $(G',T',A',B')$ such that $|R'|<|S'|$. By Claim \ref{claim:spannintreeS'DonneS}, $R = g(R')$ is such that $R\subseteq E\setminus B$,  $A\cup R\subseteq E$, $(V(A\cup R),A\cup R)$ is connected and $T\subseteq V(A\cup R)$. By remark \ref{rmdc:tailleEnsembleParg}, $|R'| = |R|$ and by remark \ref{rmdc:tailleEnsembleParf} $|S'| = |S|- |\{ \{u,x\} \mid \{u,x\},\{v,x\}\in S\}|$. Therefore, as $|R'|<|S'|$, $|R| < |S| - |\{ \{u,x\} \mid \{u,x\},\{v,x\}\in S\}|$. Hence $|R|<|S|$ a contradiction to $S$ being a solution on $G$. Hence  $S'$ is solution on $(G',T',A',B')$. 
  \end{proofclaim}

  \begin{claim}
    If $S'$ is a solution on $(G',T',A',B')$  then $g(S')$ is a solution on $(G,T,A,B)$.
  \end{claim}
  
  \begin{proofclaim}
    Suppose that $S'$ is a solution on $(G',T',A',B')$ . By Claim \ref{claim:spannintreeS'DonneS}, $S=g(S')$ is such that $S\subseteq E\setminus B$,  $A\cup S \subseteq E$, $(V(A\cup S),A\cup S)$ is connected and $T\subseteq V(A\cup S)$. We still have to prove that $S$ is of minimum size. 
    
    By contradiction, let $R$ be a solution on $(G,T,A,B)$ such that $|R|<|S|$. By Claim \ref{claim:spannintreeSDonneS'}, $R' = f(R)$ is such that $R'\subseteq E'\setminus B'$,   $A'\cup R'\subseteq E'$, $(V(A'\cup R'),A'\cup R')$ is connected and $T'\subseteq V(A'\cup R')$. By remark \ref{rmdc:tailleEnsembleParf}, $|R'| = |R|- |\{ \{u,x\} \mid \{u,x\},\{v,x\}\in R\}|$ and by remark \ref{rmdc:tailleEnsembleParg} $|S| = |S'|$. Therefore, as $|S|>|R|$, $|R'| < |S'| - |\{ \{u,x\} \mid \{u,x\},\{v,x\}\in S\}|$. Hence $|R|<|S|$ a contradiction to $S'$ being a solution on $G$. Hence  $S$ is solution on $(G,T,A,B)$. \end{proofclaim}

  As long as $A\neq \emptyset$ or $B\neq \emptyset$, Reductions 1 and 2 can be applied iteratively. Hence, we obtain an equivalent instance $I' = (G',T',\emptyset, \emptyset)$. We conclude by applying Lemma~\ref{co:steiner}.

  Observe that each reduction is applied in linear time. Reduction 1 and 2 decrease $|B|$ by at least one. Since $A\subseteq B$, there are at most $|B|<n^2$ reductions. By Lemma~\ref{co:steiner}, the final step is done in $c^k n^{O(1)}$. Hence, the problem is solvable in $O(n^2) + c^k n^{O(1)}$. This concludes the proof.
\end{proof}

Combining Theorem~\ref{th:main} with Lemma~\ref{lemma:mst} and Lemma~\ref{co:steinercomp}, we obtain the following Theorem.
\begin{theorem}
  The \textsc{\texttt{dist}-Diverse-($\Pi_1$,$\ldots$, $\Pi_r$)} problem where
  \begin{itemize}
  \item $\Pi_1= \ldots = \Pi_r = \Pi_{mst}$ and 
  \item \texttt{dist} is one of \texttt{disjoint}, \texttt{SumHam}, \texttt{MinHam}, \texttt{SumJac}, \texttt{MinJac}, \texttt{SumOts}, and \texttt{MinOts}
  \end{itemize}
  can be solved in time $(c^k\cdot n^{O(1)} + 2^{r \cdot k}) \cdot 2^{r^2 \cdot k^2} \cdot n^{O(1)}$ on an instance $(I = ((G,T),\ldots, (G,T)),d)$ where $c > 2$ and $k$ is the size of a solution of \textsc{Minimum Steiner Tree} on $(G,T)$.
\end{theorem}

\paragraph{Discussion.}
We want to highlight that if the presented result is provided where the $r$ solutions we want are solutions of the same problem with the same input for readability, one can ask for instance, the set of terminals to be different for each problem $\Pi_i$ involved.

\section{Conclusion}
\label{sec:conclusion}
In this paper, we present a meta-algorithm for finding diverse solutions to a problem parameterized by the size of the requested solutions and provide a few examples of applications. 
While this approach is quite generic, we still imposed, for most of the proposed distances, that the problem is a fixed size property.
We leave it as an open question whether this constraint is necessary or can be weakened.

{\small
  \bibliographystyle{abbrv}
  \bibliography{biblio}

@article{FuKeMoRiRoWa2007,
  author       = {Bernhard Fuchs and
                  Walter Kern and
                  Daniel M{\"{o}}lle and
                  Stefan Richter and
                  Peter Rossmanith and
                  Xinhui Wang},
  title        = {Dynamic Programming for Minimum Steiner Trees},
  journal      = {Theory Comput. Syst.},
  volume       = {41},
  number       = {3},
  pages        = {493--500},
  year         = {2007},
  url          = {https://doi.org/10.1007/s00224-007-1324-4},
  doi          = {10.1007/S00224-007-1324-4},
  bibsource    = {dblp computer science bibliography, https://dblp.org}
}

@Book{CyFoKoLoMaPiPiSa2015,
  author = 	 {{Cygan}, {Marek} and {Fomin}, {Fedor V.} and {Kowalik}, {Lukasz} and {Lokshtanov}, {Daniel} and {Marx}, {Dániel} and {Pilipczuk}, {Marcin} and {Pilipczuk}, {Michal} and {Saurabh}, {Saket}},
  title = 	 {{Parameterized Algorithms}},
  Publisher =    {Springer},
  year = 	 {2015},
}

@Article{DrMa2024,
  author = 	 {{Drabik}, {Karolina} and {Masařík}, {Tomáš}},
  title = 	 {{Finding Diverse Solutions Parameterized by Cliquewidth}},
  journal =      {CoRR},
  year = 	 {2024},
  doi = 	 {10.48550/arXiv.2405.20931},
}

@Article{fvBaFeJaMaOlPhRo2022,
  author = 	 {{Baste}, {Julien} and {Fellows}, {Michael R.} and {Jaffke}, {Lars} and {Masarik}, {Tomas} and {de Oliveira Oliveira}, {Mateus} and {Philip}, {Geevarghese} and {Rosamond}, {Frances A.}},
  title = 	 {{Diversity of solutions: An exploration through the lens of fixed-parameter tractability theory}},
  journal =      {Artificial Intelligence},
  year = 	 {2022},
  doi = 	 {10.1016/j.artint.2021.103644},
  volume = 	 {303},
  pages = 	 {103644},
}

@article{CrRo2002,
title = {On the Hamming distance of constraint satisfaction problems},
journal = {Theoretical Computer Science},
volume = {288},
number = {1},
pages = {85-100},
year = {2002},
doi = {10.1016/S0304-3975(01)00146-3},
author = {P. Crescenzi and G. Rossi},
}

@InProceedings{AuBeGoLiSr2025,
  author =	{Austrin, Per and Bercea, Ioana O. and Goswami, Mayank and Limaye, Nutan and Srinivasan, Adarsh},
  title =	{{Algorithms for the Diverse-k-SAT Problem: The Geometry of Satisfying Assignments}},
  booktitle =	{52nd International Colloquium on Automata, Languages, and Programming (ICALP 2025)},
  pages =	{14:1--14:17},
  series =	{Leibniz International Proceedings in Informatics (LIPIcs)},
  ISBN =	{978-3-95977-372-0},
  ISSN =	{1868-8969},
  year =	{2025},
  volume =	{334},
  editor =	{Censor-Hillel, Keren and Grandoni, Fabrizio and Ouaknine, Jo\"{e}l and Puppis, Gabriele},
  URL =		{https://drops.dagstuhl.de/entities/document/10.4230/LIPIcs.ICALP.2025.14},
  doi =		{10.4230/LIPIcs.ICALP.2025.14},
}

@Article{EiErErFi2013,
  author = 	 {{Eiter}, {Thomas} and {Erdem}, {Esra} and {Erdogan}, {Halit} and {Fink}, {Michael}},
  title = 	 {{Finding similar/diverse solutions in answer set programming}},
  journal =      {Theory and Practice of Logic Programming},
  year = 	 {2013},
  doi = 	 {10.1017/S1471068411000548},
  volume = 	 {13},
  number = 	 {3},
  pages = 	 {303--359},
}

@Article{ArFeLoOlWo2021,
  author = 	 {{Arrighi}, {Emmanuel} and {Fernau}, {Henning} and {Lokshtanov}, {Daniel} and {de Oliveira Oliveira}, {Mateus} and {Wolf}, {Petra}},
  title = 	 {{Diversity in Kemeny Rank Aggregation: A Parameterized Approach}},
  journal =      {Proceedings of the Thirtieth International Joint Conference on Artificial Intelligence},
  year = 	 {2021},
  doi = 	 {https://doi.org/10.24963/ijcai.2021/2},
  pages = 	 {10-16},
}

@Article{NeBoNe2021,
  author = 	 {{Neumann}, {Aneta} and {Bossek}, {Jakob} and {Neumann}, {Frank}},
  title = 	 {{Diversifying greedy sampling and evolutionary diversity optimisation for constrained monotone submodular functions}},
  journal =      {Proceedings of the Genetic and Evolutionary Computation Conference},
  year = 	 {2021},
  doi = 	 {10.1145/3449639.3459385},
  pages = 	 {261-–269},
}

@InProceedings{DoGuNeNe2023,
  author = 	 {{Do}, {Anh Viet} and {Guo}, {Mingyu} and {Neumann}, {Aneta} and {Neumann}, {Frank}},
  title = 	 {{Diverse Approximations for Monotone Submodular Maximization Problems with a Matroid Constraint}},
  booktitle = 	 {Proceedings of the Thirty-Second International Joint Conference on Artificial Intelligence (IJCAI)},
  year = 	 {2023},
  doi = 	 {10.24963/ijcai.2023/617},
  pages = 	 {5558--5566},
}

@InProceedings{ArFeOlWo2023,
  author = 	 {{Arrighi}, {Emmanuel } and {Fernau}, {Henning } and {de Oliveira Oliveira }, {Mateus } and {Wolf}, {Petra}},
  title = 	 {{Synchronization and Diversity of Solutions}},
  booktitle = 	 {Proceedings of the AAAI Conference on Artificial Intelligence (AAAI)},
  year = 	 {2023},
  doi = 	 {https://doi.org/10.1609/aaai.v37i10.26361},
  pages = 	 {11516--11524},
  number = 	 {10},
  series = 	 {37},
}

@InProceedings{NaSaSi2011,
  author = 	 {{Nadel}, {Alexander} and {Sakallah}, {Karem A.} and {Simon}, {Laurent}},
  title = 	 {{Generating Diverse Solutions in SAT}},
  booktitle = 	 {Theory and Applications of Satisfiability Testing (SAT)},
  year = 	 {2011},
  doi = 	 {10.1007/978-3-642-21581-0_23},
  pages = 	 {287--301},
  series = 	 {LNCS 6695},
}

@Article{AhMeTr2024,
  author = 	 {{Ahanor}, {Izuwa} and {Medal}, {Hugh} and {Trapp}, {Andrew C.}},
  title = 	 {{DiversiTree: A New Method to Efficiently Compute Diverse Sets of Near-Optimal Solutions to Mixed-Integer Optimization Problems}},
  journal =      {INFORMS Journal on Computing},
  year = 	 {2023},
  doi = 	 {10.1287/ijoc.2022.0164},
  volume = 	 {36},
  number = 	 {1},
  pages = 	 {61--77},
}

@Article{DaWo2009,
  author = 	 {{Danna}, {Emilie} and {Woodruff}, {David L.}},
  title = 	 {{How to select a small set of diverse solutions to mixed integer programming problems}},
  journal =      {Operations Research Letters},
  year = 	 {2009},
  doi = 	 {10.1016/j.orl.2009.03.004},
  volume = 	 {37},
  number = 	 {4},
  pages = 	 {255--260},
}

@InProceedings{PeTr2015,
  author = 	 {{Petit}, {Thierry} and {Trapp}, {Andrew C.}},
  title = 	 {{Finding Diverse Solutions of High Quality to Constraint Optimization Problems}},
  booktitle = 	 {Proceedings of the Twenty-Fourth International Joint Conference on Artificial Intelligence (IJCAI)},
  year = 	 {2015},
  pages = 	 {260--267},
}

@InProceedings{HeHnSuWa2005,
  author = 	 {{Hebrard}, {Emmanuel} and {Hnich}, {Brahim} and {O'Sullivan}, {Barry} and {Walsh}, {Toby}},
  title = 	 {{Finding diverse and similar solutions in constraint programming}},
  booktitle = 	 {The 20th AAAI Conference on Artificial Intelligence (AAAI)},
  year = 	 {2005},
  pages = 	 {372--377},
}

@InProceedings{GaBePhSc2018,
  author = 	 {{Gabor}, {Thomas} and {Belzner}, {Lenz} and {Phan}, {Thomy} and {Schmid}, {Kyrill}},
  title = 	 {{Preparing for the Unexpected: Diversity Improves Planning Resilience in Evolutionary Algorithms}},
  booktitle = 	 {2018 IEEE International Conference on Autonomic Computing (ICAC)},
  year = 	 {2018},
  doi = 	 {10.1109/ICAC.2018.00023},
  pages = 	 {131--140},
}

@InProceedings{WiOp2003,
  author = 	 {{Wineberg}, {Mark} and {Oppacher}, {Franz}},
  title = 	 {{The Underlying Similarity of Diversity Measures Used in Evolutionary Computation}},
  booktitle = 	 {Genetic and Evolutionary Computation (GECCO)},
  year = 	 {2003},
  doi = 	 {10.1007/3-540-45110-2_21},
  pages = 	 {1493--1504},
  series = 	 {2724},
}

@InProceedings{FoGoJaPhSa2020,
  author = 	 {{Fomin}, {Fedor V.} and {Golovach}, {Petr A.} and {Jaffke}, {Lars} and {Philip}, {Geevarghese} and {Sagunov}, {Danil}},
  title = 	 {{Diverse Pairs of Matchings}},
  booktitle = 	 {Proceedings of the 31st International Symposium on Algorithms and Computation (ISAAC)},
  year = 	 {2020},
  doi = 	 {10.4230/LIPIcs.ISAAC.2020.26},
  pages = 	 {26:1--26:12},
  series = 	 {LIPIcs 181},
}

@InProceedings{FoGoPaPhSa2021,
  author = 	 {{Fomin}, {Fedor V.} and {Golovach}, {Petr A.} and {Panolan}, {Fahad} and {Philip}, {Geevarghese} and {Saurabh}, {Saket}},
  title = 	 {{Diverse Collections in Matroids and Graphs}},
  booktitle = 	 {Proceedings of the 38th International Symposium on Theoretical Aspects of Computer Science (STACS)},
  year = 	 {2021},
  doi = 	 {10.4230/LIPIcs.STACS.2021.31},
  pages = 	 {31:1--31:14},
  series = 	 {LIPIcs 187},
}

@Article{IcIw2024,
  author = 	 {{Ichikawa}, {Yuma} and {Iwashita}, {Hiroaki}},
  title = 	 {{Continuous Tensor Relaxation for Finding Diverse Solutions in Combinatorial Optimization Problems}},
  journal =      {CoRR},
  year = 	 {2024},
  doi = 	 {10.48550/arXiv.2402.02190},
}

@InProceedings{InGaStTa2020,
  author = 	 {{Ingmar}, {Linnea} and {Garcia de la Banda}, {Maria} and {Stuckey}, {Peter J.} and {Tack}, {Guido}},
  title = 	 {{Modelling Diversity of Solutions}},
  booktitle = 	 {Proceedings of the AAAI Conference on Artificial Intelligence (AAAI)},
  year = 	 {2020},
  doi = 	 {10.1609/aaai.v34i02.5512},
  pages = 	 {1528-1535},
  number = 	 {2},
  series = 	 {34},
}

@InProceedings{HaKoKuLeOt2022,
  author = 	 {{Hanaka}, {Tesshu} and {Kobayashi}, {Yasuaki} and {Kurita}, {Kazuhiro} and {Lee}, {See Woo} and {Otachi}, {Yota}},
  title = 	 {{Computing Diverse Shortest Paths Efficiently: A Theoretical and Experimental Study}},
  booktitle = 	 {Proceedings of the AAAI Conference on Artificial Intelligence (AAAI)},
  year = 	 {2022},
  doi = 	 {10.1609/aaai.v36i4.20290},
  pages = 	 {3758-3766},
  number = 	 {4},
  series = 	 {36},
}

@InProceedings{HaKoKuOt2021,
  author = 	 {{Hanaka}, {Tesshu} and {Kobayashi}, {Yasuaki} and {Kurita}, {Kazuhiro} and {Otachi}, {Yota}},
  title = 	 {{Finding Diverse Trees, Paths, and More}},
  booktitle = 	 {Proceedings of the AAAI Conference on Artificial Intelligence (AAAI)},
  year = 	 {2021},
  doi = 	 {10.1609/aaai.v35i5.16495},
  pages = 	 {3778-3786},
  number = 	 {5},
  series = 	 {35},
}
}

\end{document}